\documentclass[11pt,letterpaper,reqno]{amsart}
\usepackage{amsmath,amsthm,amsfonts,amssymb,mathtools}

\theoremstyle{theorem}
\newtheorem{theorem}{Theorem}
\newtheorem{proposition}{Proposition}

\newtheorem*{problem}{Problem}
\newtheorem{openproblem}{Open problem}

\theoremstyle{definition}

\begin{document}

\title[From electrostatic potentials to yet another triangle center]{From electrostatic potentials to yet another triangle center}

\author{Hrvoje Abraham}
\address{Hrvoje Abraham, Artes Calculi, Deren\v{c}inova 1, 10000 Zagreb, Croatia}
\email{hrvoje.abraham@artcalc.com}
\email{ahrvoje@gmail.com}

\author{Vjekoslav Kova\v{c}}
\address{Vjekoslav Kova\v{c}, University of Zagreb, Faculty of Science, Department of Mathematics, Bijeni\v{c}ka cesta 30, 10000 Zagreb, Croatia}
\email{vjekovac@math.hr}

\begin{abstract}
We study the problem of finding a point of maximal electrostatic potential inside an arbitrary triangle with homogeneous surface charge distribution.
In this article we derive several synthetic and analytic relations for its location in the plane.
Moreover, this point satisfies the definition of a triangle center, different from any of previously discovered centers in Clark Kimberling's encyclopedia.
\end{abstract}


\date{February 26, 2014}

\maketitle

\section{Introduction}

The topic we are about to discuss was initiated by a concrete and practical question in physics
that has eventually revealed its unexpectedly interesting geometrical flavor.
Let us begin with a statement of this theoretical problem and postpone applied motivation to the end of this section.

\begin{problem}
Suppose that a planar triangle $T$ is a continuous source of charge, which is homogeneously distributed over its surface,
i.e.\@ the charge density is constant over the triangle.
At which point in the same plane the electrostatic potential of $T$ attains its maximum value?
\end{problem}

All physical notions will be accompanied with their precise definitions and the discussion will soon turn into elementary geometrical considerations.
Let us recall that the potential of a point source with charge $q$ evaluated at a point that is $r$ units apart is given by $V(r)=k q/r$.
This is merely a restatement of \emph{Coulomb's law} and the constant $k$ is not important for us.
By ``superposition principle'' for multiple charges it is therefore reasonable to define the potential generated by the whole triangle $T$ as
\begin{equation}
\label{formulapotential}
V(P) = \iint_{T} \frac{d\lambda(Q)}{|PQ|}
\end{equation}
for any point $P$ in the plane.
Here $\lambda$ denotes the two-dimensional Lebesgue measure (i.e.\@ the area measure),
$Q$ is an integration variable, and $|PQ|$ denotes the distance between points $P$ and $Q$.
We are careless about the multiplicative constant or the charge density and we even omit them from writing.
In Cartesian coordinates the above formula becomes simply
$$ V(x,y) = \iint_{T} \frac{dx'dy'}{\sqrt{(x'-x)^2+(y'-y)^2}} . $$

It is easy to see that $V$ is indeed a well-defined function on the whole plane.
One can draw contour graphs of \eqref{formulapotential} for various choices of triangles
using the Mathematica command ContourPlot \cite{m9} and the level sets will look as those in figure \ref{figurepotential}.
\begin{figure}[ht]
\begin{center}
\includegraphics[width=2.5in]{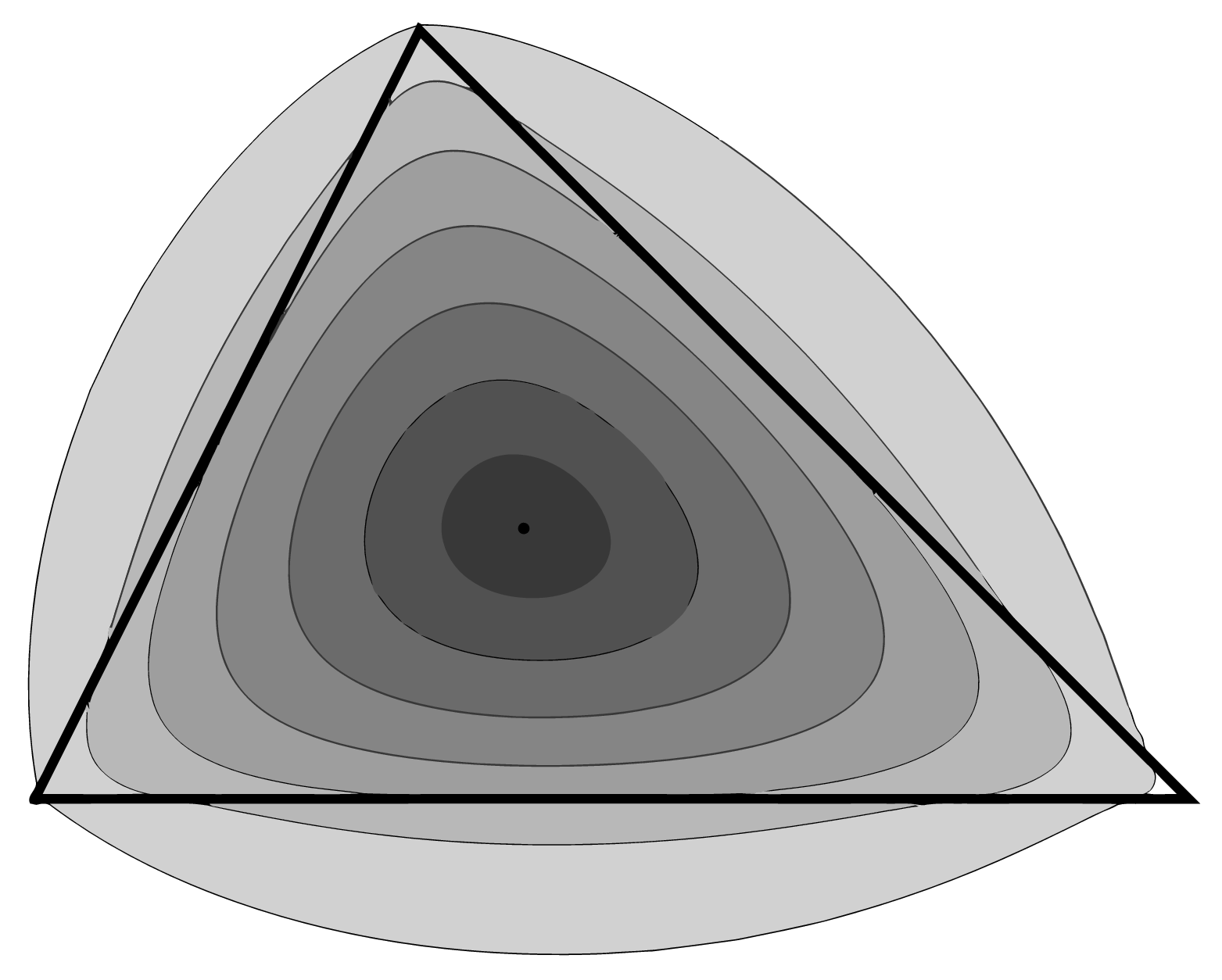}
\end{center}
\caption{Contour graph of $V$.}
\label{figurepotential}
\end{figure}
Such drawings can make us suspect that $V$ has the shape of a single ``mountain peak,'' but this certainly could not pass as a rigorous argument.
It is not immediately clear from the formula that there even exist a point $P_\textup{max}$ inside $T$ where $V$ attains its maximum
and it is certainly not obvious that such point should be unique for every triangle.
Moreover, we would like to locate this point, in a certain sense, for an arbitrary given triangle $T$.

What would be the physical meaning of the maximum potential point?
It is the point where the electrostatic field $\vec{E}$ generated by $T$ stabilizes.
Let us perform a simple thought experiment. Assume that $T$ is charged positively and place a negative point charge somewhere in the plane.
It will necessarily be driven by electrostatic forces unless it is placed at a point where it ``feels perfectly stable.''
Figure \ref{figurefield} illustrates several integral curves of the vector field $\vec{E}$,
which are also known in physics as \emph{lines of force} or \emph{field lines}.
\begin{figure}[ht]
\begin{center}
\includegraphics[width=2.5in]{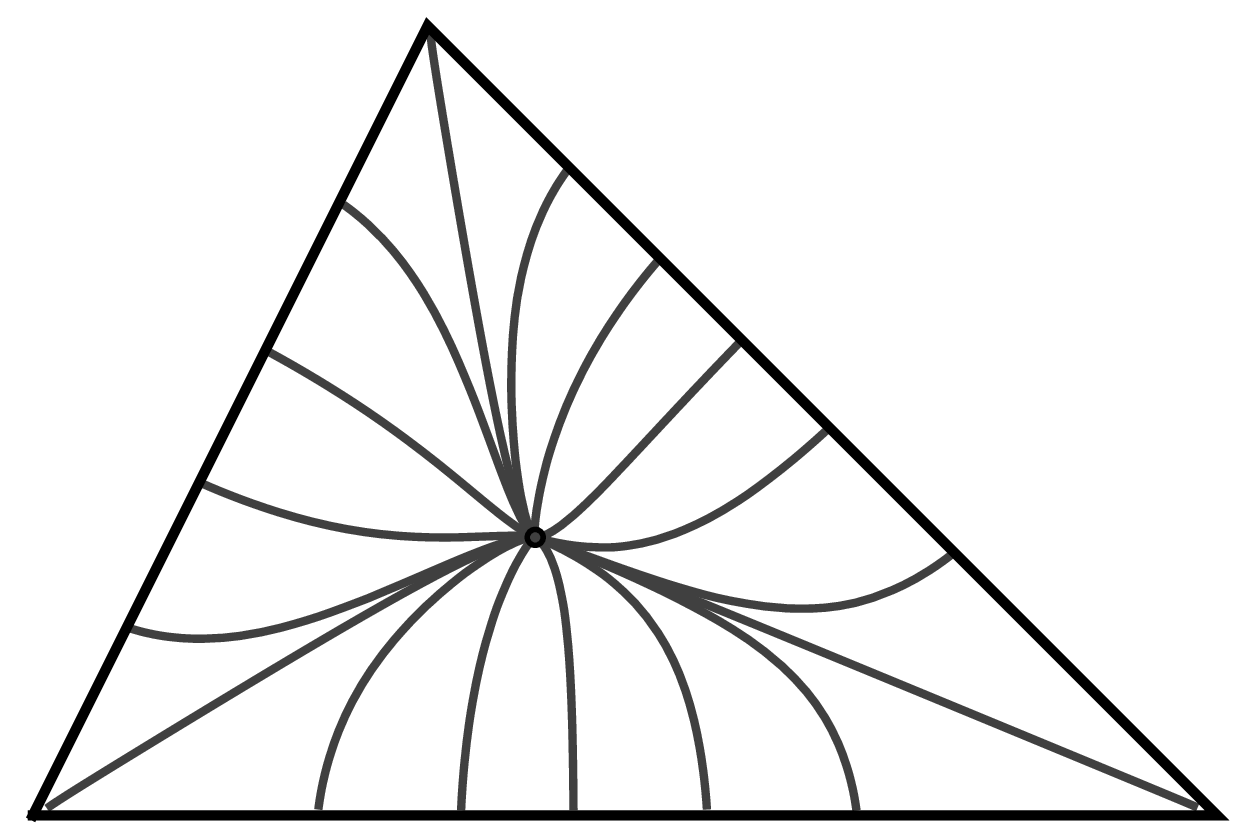}
\end{center}
\caption{Integral curves of $\vec{E}$ that lie inside $T$.}
\label{figurefield}
\end{figure}
Observe that they all meet at the same point inside $T$. This experiment is once again very far from a rigorous proof.
Existence and uniqueness of the maximum potential point will follow from more general results in convex analysis and will be discussed in the next section.
However, in the case of a triangle, they will also come as a byproduct of our attempts to specify its location throughout the rest of the paper.

We have just mentioned the notion of electrostatic field, so what would that field be in the case of our charged triangle $T$?
It can be defined simply as $\vec{E}=-\nabla V$ at any point where the potential $V$ is differentiable.
In physics, the electric field is sometimes (but not always) given before the potential.
We have intentionally ordered things this way, simply because the potential of $T$ was easier to define mathematically.
Going back to a point source, an easily derived and well-known formula is $\vec{E}=k q \vec{r}/r^3$.
Here $\vec{r}$ denotes a directed line segment from the source to a point where the field is computed.
Using the superposition principle once again we suspect that the correct corresponding expression is
\begin{equation}
\label{formulafield1}
\vec{E}(P) = \iint_{T} \frac{\overrightarrow{QP}}{|PQ|^3}d\lambda(Q)
= - \iint_{T} \frac{\overrightarrow{PQ}}{|PQ|^3}d\lambda(Q) ,
\end{equation}
or coordinate-wise with $\vec{i}$ and $\vec{j}$ being the standard unit vectors,
$$ \vec{E}(x,y) = - \iint_{T} \frac{(x'-x)\vec{i}+(y'-y)\vec{j}}{((x'-x)^2+(y'-y)^2)^{3/2}}dx'dy' . $$
However, the double integral in the above formula will not be absolutely convergent unless $P=(x,y)$ lies outside $T$.
To explain the difficulty, assume that $P$ is contained in the triangle interior,
together with a ``small'' disk $\textup{D}_{\varepsilon}(P)$ of radius $\varepsilon$ around it.
We insert absolute values inside the double integral and only integrate over this disk.
Changing to a polar coordinate system centered at $P$ we obtain
$$ \iint_{\textup{D}_{\varepsilon}(P)} \frac{|\overrightarrow{PQ}|}{|PQ|^3}d\lambda(Q)
= \int_{0}^{\varepsilon}\!\int_{0}^{2\pi}\! \frac{r}{r^3} r dr d\varphi = +\infty , $$
because $\int_{0}^{\varepsilon}dr/r$ diverges.

In order to get a valid formula for $\vec{E}$ that would hold for points $P$ in the interior of $T$,
one simply has to observe that the contributions $\frac{\overrightarrow{PQ}}{|PQ|^3}$ of points $Q\in\textup{D}_{\varepsilon}(P)$
cancel out each other completely, see figure \ref{figurecancel}.
\begin{figure}[ht]
\begin{center}
\includegraphics[width=2.59in]{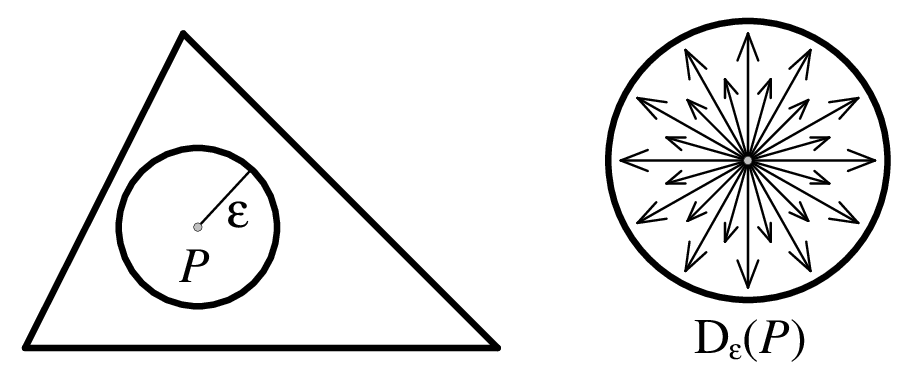}
\end{center}
\caption{Cancellations in the singular part of the integral.}
\label{figurecancel}
\end{figure}
Therefore,
\begin{equation}
\label{formulafield3}
\vec{E}(P) = - \iint_{T\setminus\textup{D}_{\varepsilon}(P)} \frac{\overrightarrow{PQ}}{|PQ|^3}d\lambda(Q)
\end{equation}
should hold for $P$ inside the triangle. Indeed, one can even let $\varepsilon\to 0$,
obtaining the expression called the \emph{principal value} of the integral:
\begin{equation}
\label{formulafield2}
\vec{E}(P) = - \textup{p.v.}\!\iint_{T} \frac{\overrightarrow{PQ}}{|PQ|^3}d\lambda(Q) .
\end{equation}
Things remain problematic for points $P$ at the boundary, because the same argument shows that the expression for $\vec{E}(P)$
does not converge in any usual sense. Indeed, the potential is continuous but not differentiable at those points.

The main source of motivation for the problem comes from implementation of a certain type of \emph{boundary element method (BEM)} for electrostatic problems
\cite{b77}, \cite{j63}, \cite{lsa06}, \cite{r67}.
Boundary element methods are usually formulated by surface elements of a three-dimensional object
and these elements are in turn most often represented by triangles.
In the case of an electrostatic problem, a single triangle potential could be evaluated either at vertices,
or at a certain interior point, depending on the formulation of the method.
In the later case, it is common to take the center of mass (i.e.\@ the centroid), but there is no reason or evidence why this would be the best choice.
Indeed, one can argue that using the maximum potential point provides better results, but such discourse is out of the scope of this paper.
Calculating its coordinates and discovering its properties proved to be challenges on their own.

\section{Existence and uniqueness}
\label{sectionproperties}

In this section we start the rigorous mathematical treatment of the problem.
Our potential is a particular instance of the so-called \emph{fractional integral},
\begin{equation}\label{formulariesz}
(I_p f)(x,y) = \iint \Big((x'-x)^2+(y'-y)^2\Big)^{p/2} f(x',y')\,dx'dy' ,
\end{equation}
which is also known as the \emph{Riesz potential} \cite{r49} when $-2<p<0$ and when it is properly normalized.
In order to obtain \eqref{formulapotential}, one only has to take $p=-1$ and choose $f$ to be the indicator function of $T$.

Extreme points of ``regularized'' versions of $I_p f$ when $p$ is a real number and $f$ is the characteristic function of a general convex set
(even in higher dimensions) have already been studied in the literature.
They were named \emph{radial centers} by M. Moszy\'{n}ska \cite{m00}, who seems to be the first to establish their existence and uniqueness
for $-2<p\leq 1$, while the remaining cases were studied by I. Herburt \cite{h07} and J. O'Hara \cite{oh12},
who called these points \emph{$r^p$ centers}.
Herburt, Moszy\'{n}ska, and Peradzy\'{n}ski \cite{hmp05} gave physical interpretations of radial centers,
mentioning gravitational and electrostatic potentials for $p=-1$, but do not specialize the discussion to triangles.
On the other hand, we need to mention an unpublished text by K. Shibata \cite{s09} on a similarly defined but different point in a triangle,
corresponding to $p=-2$, which we discuss briefly in the last section.

As we have already said, existence and uniqueness of the maximum point for $V$ follows
from general results of Moszy\'{n}ska \cite[Section 3]{m00} for general compacts convex sets $T$ with nonempty interior.
Moreover, Herburt \cite{h08} showed that the maximum point lies in the interior of $T$ if the set $T$ has piecewise smooth boundary.
However, since we are only interested in a very special case when $T$ is a triangle in $\mathbb{R}^2$,
we are able to reprove these facts easily between the lines of the more precise results on the maximum point location.
This keeps the material elementary and self-contained.

The following proposition is an easy exercise in vector calculus, so we only provide proofs of its nontrivial parts.

\begin{proposition}\label{prop1}
\item[(a)]
Potential $V$ is finite and continuous on the whole plane.
\item[(b)]
$V(P)\to 0$ uniformly as the distance from $P$ to $T$ tends to $\infty$.
\item[(c)]
Potential $V$ is differentiable both in the interior and in the exterior of $T$.
\item[(d)]
Field $\vec{E}=-\nabla V$ is given by \eqref{formulafield1} for exterior points $P$ and by \eqref{formulafield3} or
\eqref{formulafield2} for points $P$ in the interior of $T$.
\item[(e)]
Potential $V$ cannot attain local maxima in the exterior or on the boundary of $T$.
\end{proposition}

\begin{proof}[Proof of proposition \ref{prop1}]
\emph{Parts (a) and (b)} are very easy and follow simply from absolute integrability of the function in \eqref{formulapotential}
and boundedness of the domain $T$.

\emph{Parts (c) and (d).}
Fix a point $P_0$ inside $T$ and choose $\varepsilon>0$ twice smaller than the distance from $P_0$ to the boundary of $T$.
We need to show that $V$ is differentiable at $P_0$ and that $\nabla V(P_0)=-\vec{E}(P_0)$, where $\vec{E}(P_0)$ is given by formula \eqref{formulafield3}.
Take any point $P$ such that $|PP_0|<\varepsilon$.
Parts of the integrals in the expression $V(P)-V(P_0)$
corresponding to $\textup{D}_{\varepsilon}(P_0)\!\cap\!\textup{D}_{\varepsilon}(P)$ cancel out by symmetry,
so this difference is equal to
$$ \iint_{T\setminus(\textup{D}_{\varepsilon}(P_0)\cup\textup{D}_{\varepsilon}(P))}\Big(\frac{1}{|PQ|} - \frac{1}{|P_0 Q|}\Big) d\lambda(Q) . $$
Using
$$ |P_0 Q|^2 - |P Q|^2 = 2\overrightarrow{P_0 Q}\cdot\overrightarrow{P_0 P}-|P_0 P|^2 , $$
it can be rewritten as
$$ V(P) - V(P_0) = \iint_{T\setminus(\textup{D}_{\varepsilon}(P_0)\cup\textup{D}_{\varepsilon}(P))}
\frac{2\overrightarrow{P_0 Q}\cdot\overrightarrow{P_0 P}-|P_0 P|^2}{|P_0 Q| |PQ| (|P_0 Q|+|PQ|)} d\lambda(Q) . $$
On the other hand, from \eqref{formulafield3},
$$ \vec{E}(P_0)\cdot \overrightarrow{P_0 P} = -\iint_{T\setminus\textup{D}_{\varepsilon}(P_0)}
\frac{\overrightarrow{P_0 Q}\cdot\overrightarrow{P_0 P}}{|P_0 Q|^3} d\lambda(Q) . $$
After simple algebraic manipulations
and by splitting
$$ \textup{D}_{\varepsilon}(P_0) = \big(\textup{D}_{\varepsilon}(P_0)\cup\textup{D}_{\varepsilon}(P)\big)
\setminus \big(\textup{D}_{\varepsilon}(P)\setminus\textup{D}_{\varepsilon}(P_0)\big) , $$
we arrive at
$$ \frac{1}{|P_0 P|}\Big(V(P)-V(P_0)+\vec{E}(P_0)\cdot \overrightarrow{P_0 P}\Big) = J_1 - J_2 - J_3 , $$
where
\begin{align*}
J_1 & = \iint_{T\setminus(\textup{D}_{\varepsilon}(P_0)\cup\textup{D}_{\varepsilon}(P))}
\frac{\overrightarrow{P_0 Q}\cdot\overrightarrow{P_0 P}}{|P_0 Q| |P_0 P|}
\,\frac{2|P_0 Q|+|PQ|}{|P_0 Q|+|PQ|} \,\frac{|P_0 Q|-|PQ|}{|P_0 Q|^2 |PQ|} \,d\lambda(Q) , \\
J_2 & = \iint_{T\setminus(\textup{D}_{\varepsilon}(P_0)\cup\textup{D}_{\varepsilon}(P))}
\frac{|P_0 P|}{|P_0 Q| |PQ| (|P_0 Q|+|PQ|)} \,d\lambda(Q) , \\
J_3 & = \iint_{\textup{D}_{\varepsilon}(P)\setminus\textup{D}_{\varepsilon}(P_0)}
\frac{\overrightarrow{P_0 Q}\cdot\overrightarrow{P_0 P}}{|P_0 Q| |P_0 P|} \,\frac{1}{|P_0 Q|^2} \,d\lambda(Q) .
\end{align*}
Using $|P_0 Q|\geq\varepsilon$, $|PQ|\geq\varepsilon$, and $\big||P_0 Q|-|PQ|\big|\leq|P_0 P|$ the first integral is easily bounded as
$$ |J_1| \leq \frac{2}{\varepsilon^3}\lambda(T)|P_0 P| $$
and similarly we get
$$ |J_2| \leq \frac{1}{2\varepsilon^3}\lambda(T)|P_0 P| ,
\quad |J_3| \leq \frac{1}{\varepsilon^2}\lambda\big(\textup{D}_{\varepsilon}(P)\!\setminus\!\textup{D}_{\varepsilon}(P_0)\big) . $$
Letting $P\to P_0$ we conclude
$$ \lim_{P\to P_0} \frac{V(P)-V(P_0)+\vec{E}(P_0)\cdot \overrightarrow{P_0 P}}{|P_0 P|} =0 , $$
which is precisely what we needed.

For points $P_0$ in the exterior of $T$ the proof can follow the same lines.
Moreover, an even shorter proof can be given for such $P_0$ by entirely standard arguments of interchanging limits and integrals,
as the integral in \eqref{formulafield1} is an absolutely convergent one.

\emph{Part (e).}
Begin by taking a point $P_0$ outside $T$. Informally saying, the field does not vanish at $P_0$ since it has to ``point'' away from $T$.
More rigorously, let $l$ be any line passing though $P_0$ and containing $T$ entirely in one of the two corresponding half-planes,
see figure \ref{figureseparating}.
\begin{figure}[ht]
\begin{center}
\includegraphics[width=2.1in]{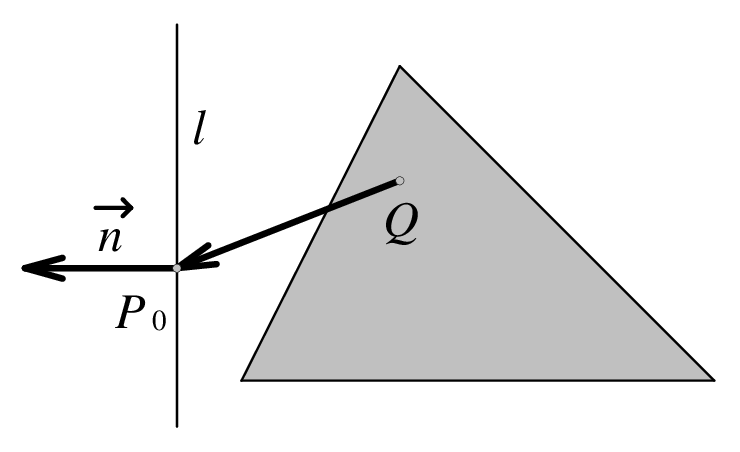}
\end{center}
\caption{Treatment of exterior points.}
\label{figureseparating}
\end{figure}
If $\vec{n}$ is a vector normal to $l$ and oriented in the opposite direction, then formula \eqref{formulafield1} yields
$$ \vec{E}(P_0)\cdot\vec{n} = \iint_{T}\frac{\overrightarrow{QP_0}\cdot\vec{n}}{|P_0 Q|^3}d\lambda(Q) > 0 . $$
Consequently, $(\nabla V)(P_0)\neq\vec{0}$, so $P_0$ cannot be a stationary point for $V$.

The same argument ``almost works'' for points at the triangle boundary.
Even though $\vec{E}(P_0)$ does not exist, we can imagine that it is a vector of infinite length pointing outwards.
The reader can modify the proof of parts (c) and (d) to show that
$$ \lim_{h\to 0}\frac{V(P_0-h\vec{n})-V(P_0)}{h} = +\infty $$
holds for the same choice of $\vec{n}$. Once again, we conclude that $P_0$ is not a local maximum point for $V$.
\end{proof}

It is now easy to conclude that potential $V$ attains its maximum at some point inside triangle $T$
and at each such point $P$ one has $\vec{E}(P)=\vec{0}$.
Indeed, by positivity and parts (a) and (b) of Proposition \ref{prop1} it follows that $V$ is bounded and has a maximum
that is attained at some (finite) point in the plane.
By part (e) we know that any such point must lie in the interior of $T$.
Finally, the second assertion is a consequence of parts (c) and (d).

We need to remark that an explicit formula for $V$ can be computed, although it is rather complicated and not practically useful.
Instead, it will be more useful to transform formula \eqref{formulafield3} for $\vec{E}(P)$ in the next section.

\section{Geometric relations}
\label{sectionrelations}

Throughout this section suppose that $P$ is a stationary point inside a positively oriented triangle $T=\triangle ABC$,
i.e.\@ the corresponding vector field $\vec{E}$ vanishes at $P$.
We already know that $P$ has to coincide with the unique maximum point of $V$, but prefer to use condition $\vec{E}(P)=\vec{0}$ only,
in order to reprove the uniqueness result.
Denote its distances from vertices $A,B,C$ respectively by
$$ r_A=|PA|, \ \ r_B=|PB|, \ \ r_C=|PC| . $$
Let us also introduce convenient notation for the several angles it determines,
$$ \begin{array}{lll}
\alpha_1=\angle BAP, & \beta_1=\angle CBP, & \gamma_1=\angle ACP, \\
\alpha_2=\angle PAC, & \beta_2=\angle PBA, & \gamma_2=\angle PCB,
\end{array} $$
as in figure \ref{figurenotation}.
Finally, we use standard notation for triangle sidelengths and angles:
$$ a=|BC|, \ \ b=|CA|, \ \ c=|AB|, \ \ \alpha=\angle BAC, \ \ \beta=\angle CBA, \ \ \gamma=\angle ACB . $$
\begin{figure}[ht]
\begin{center}
\includegraphics[width=2.19in]{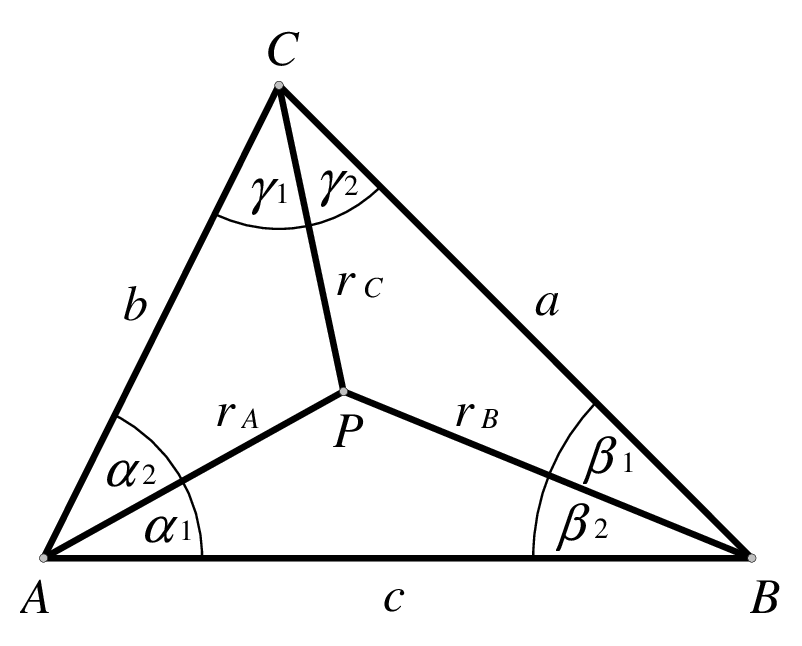}
\end{center}
\caption{Convenient notation.}
\label{figurenotation}
\end{figure}

The following theorem gives two simple relations that enable us to locate such point $P$ in the plane.
The first relation is in terms of distances from triangle vertices, while the second one is in terms of the angles defined above.

\begin{theorem}\label{thm1}
If $P$ is a point inside triangle $ABC$ such that $\vec{E}(P)=\vec{0}$, then
\begin{equation}\label{formularelation1}
\left(\frac{r_B+r_C-a}{r_B+r_C+a}\right)^{1/a} = \left(\frac{r_C+r_A-b}{r_C+r_A+b}\right)^{1/b} = \left(\frac{r_A+r_B-c}{r_A+r_B+c}\right)^{1/c}
\end{equation}
and
\begin{equation}\label{formularelation2}
\bigg(\!\tan\frac{\beta_1}{2}\tan\frac{\gamma_2}{2}\bigg)^{\frac{1}{\sin\alpha}}
= \bigg(\!\tan\frac{\gamma_1}{2}\tan\frac{\alpha_2}{2}\bigg)^{\frac{1}{\sin\beta}}
= \bigg(\!\tan\frac{\alpha_1}{2}\tan\frac{\beta_2}{2}\bigg)^{\frac{1}{\sin\gamma}} .
\end{equation}
\end{theorem}

In particular, equations \eqref{formularelation1} and \eqref{formularelation2} hold for the maximum point of potential $V$.

\begin{proof}[Proof of theorem \ref{thm1}]
Take $P$ to be the origin of the coordinate system and change to polar coordinates.
Let us denote by $M_\varphi$ the point at the intersection of the polar ray determined by an angle $\varphi\in[0,2\pi)$
with the boundary of $\triangle ABC$. Furthermore, let us write $R(\varphi)=|P M_\varphi|$.
For $\varepsilon>0$ small enough formula \eqref{formulafield3} becomes
\begin{align*}
\vec{E}(P) & = - \int_{\varepsilon}^{R(\varphi)} \!\!\int_{0}^{2\pi} \frac{r(\cos\varphi)\vec{i}+r(\sin\varphi)\vec{j}}{r^3} \,r dr d\varphi \\
& = - \int_{0}^{2\pi} \!\big(\log R(\varphi)-\log\varepsilon\big) \Big((\cos\varphi)\vec{i}+(\sin\varphi)\vec{j}\Big) d\varphi
\end{align*}
and then using $\int_{0}^{2\pi}\cos\varphi \,d\varphi=0=\int_{0}^{2\pi}\sin\varphi \,d\varphi$ we get
$$ \vec{E}(P) = - \int_{0}^{2\pi} \!\log R(\varphi) \,\Big((\cos\varphi)\vec{i}+(\sin\varphi)\vec{j}\Big) d\varphi . $$

For the rest of the proof it will be convenient to represent vectors by complex numbers, i.e.\@ to work in the complex plane.
Using $e^{i\varphi}=\cos\varphi+i\sin\varphi$ the condition $\vec{E}(P)=\vec{0}$ becomes simply
\begin{equation}\label{formulaequalszero}
\int_{0}^{2\pi} \!\log R(\varphi) \,e^{i\varphi} d\varphi = 0 .
\end{equation}

The next step is to find an expression for $\log R(\varphi)$.
Let vertices $A,B,C$ have complex coordinates
$$ r_A e^{i\varphi_A}, r_B e^{i\varphi_B}, r_C e^{i\varphi_C} $$
and let vectors $\overrightarrow{CB},\overrightarrow{AC},\overrightarrow{BA}$ be represented by complex numbers
$$ a e^{i\theta_{a}}, b e^{i\theta_{b}}, c e^{i\theta_{c}} . $$
Without loss of generality suppose that $M_\varphi$ lies on side $AB$ of $\triangle ABC$, which is the same as saying
$\varphi_A<\varphi<\varphi_B$, where we possibly need to adjust the angles by adding appropriate multiples of $2\pi$.
Let $d_c$ denote the distance from $P$ to the line $AB$ and let $\psi$ denote the angle $\angle BM_\varphi P$.
From figure \ref{figureangledistance} we see that $\psi=\varphi-\varphi_A+\alpha_1$ and $R(\varphi)=d_c/\sin\psi$, i.e.
$$ \log R(\varphi)=\log d_c-\log\sin\psi . $$
\begin{figure}[ht]
\begin{center}
\includegraphics[width=3.03in]{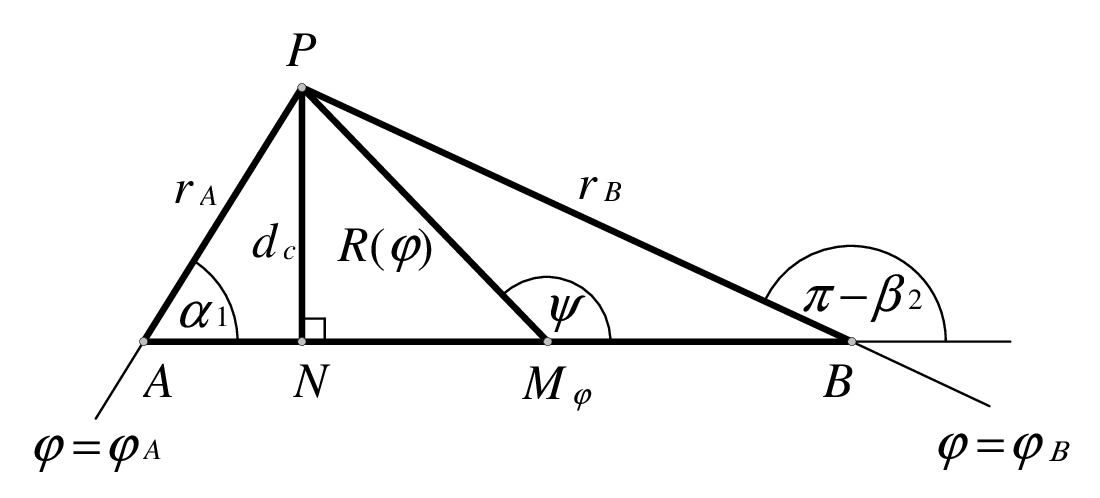}
\end{center}
\caption{Discussion of the range $\varphi_A<\varphi<\varphi_B$.}
\label{figureangledistance}
\end{figure}
Observing that $\psi$ ranges from $\alpha_1$ to $\pi-\beta_2$ we get
$$ \int_{\varphi_A}^{\varphi_B} \log R(\varphi) \, e^{i\varphi} d\varphi
= \log d_c \int_{\varphi_A}^{\varphi_B} \!e^{i\varphi} d\varphi
- \int_{\alpha_1}^{\pi-\beta_2} \!(\log\sin\psi) \,e^{i(\psi+\varphi_A-\alpha_1)} d\psi . $$

First, we use an immediate formula
\begin{equation}\label{formulaintegral1}
\int_{\eta}^{\vartheta} e^{i\varphi}d\varphi = \big(-i e^{i\varphi}\big)\Big|_{\varphi=\eta}^{\varphi=\vartheta} .
\end{equation}
Next, it is an easy exercise in integration by parts to obtain
$$ \int_{\eta}^{\vartheta} (\log\sin\psi) \cos\psi \,d\psi
= \Big(\big(\log\sin\psi-1\big)\sin\psi\Big)\Big|_{\psi=\eta}^{\psi=\vartheta} $$
and
$$ \int_{\eta}^{\vartheta} (\log\sin\psi) \sin\psi \,d\psi
= \Big(-\!\big(\log\sin\psi-1\big)\cos\psi+\log\tan{\textstyle\frac{\psi}{2}}\Big)\Big|_{\psi=\eta}^{\psi=\vartheta} $$
for angles $0<\eta<\vartheta<\pi$.
Combining we get
\begin{equation}\label{formulaintegral2}
 \int_{\eta}^{\vartheta} (\log\sin\psi) e^{i\psi} d\psi
= \Big(-i\big(\log\sin\psi-1\big)e^{i\psi}+i\log\tan{\textstyle\frac{\psi}{2}}\Big)\Big|_{\psi=\eta}^{\psi=\vartheta} .
\end{equation}
From formulas \eqref{formulaintegral1}, \eqref{formulaintegral2} we obtain
\begin{align*}
\int_{\varphi_A}^{\varphi_B} \log R(\varphi) \, e^{i\varphi} d\varphi
& = - i\log d_c \,e^{i\varphi_B} + i\log d_c \,e^{i\varphi_A} \\[-1.5mm]
& \ + i e^{i(\varphi_A-\alpha_1+\pi-\beta_2)}(\log\sin\beta_2\!-\!1) - i e^{i\varphi_A}(\log\sin\alpha_1\!-\!1) \\
& \ - i e^{i(\varphi_A-\alpha_1)}\log\tan{\textstyle\frac{\pi-\beta_2}{2}}
+ i e^{i(\varphi_A-\alpha_1)}\log\tan{\textstyle\frac{\alpha_1}{2}} ,
\end{align*}
and then using
$$ d_c/\sin\alpha_1=r_A, \ \ d_c/\sin\beta_2=r_B, \ \ \varphi_A\!-\!\alpha_1\!+\!\pi\!-\!\beta_2=\varphi_B, \ \ \varphi_A\!-\!\alpha_1=\theta_c $$
we get
\begin{align*}
\int_{\varphi_A}^{\varphi_B} \log R(\varphi) \, e^{i\varphi} d\varphi
& = - i e^{i\varphi_B}(\log r_B-1) + i e^{i\varphi_A}(\log r_A-1) \\[-1.5mm]
& \ + i e^{i\theta_c}\big(\log\tan{\textstyle\frac{\alpha_1}{2}} - \log\cot{\textstyle\frac{\beta_2}{2}}\big) .
\end{align*}
Adding this one and the two analogous relations, applying \eqref{formulaequalszero}, and observing cancellations
of $i e^{i\varphi_A}(\log r_A-1)$ and the two alike terms gives
$$ e^{i\theta_c}\log(\tan{\textstyle\frac{\alpha_1}{2}}\tan{\textstyle\frac{\beta_2}{2}})
+ e^{i\theta_a}\log(\tan{\textstyle\frac{\beta_1}{2}}\tan{\textstyle\frac{\gamma_2}{2}})
+ e^{i\theta_b}\log(\tan{\textstyle\frac{\gamma_1}{2}}\tan{\textstyle\frac{\alpha_2}{2}}) = 0 . $$
We can interpret this using vectors once again as
$$ \frac{\log(\tan\frac{\alpha_1}{2}\tan\frac{\beta_2}{2})}{c}\overrightarrow{BA}
+ \frac{\log(\tan\frac{\beta_1}{2}\tan\frac{\gamma_2}{2})}{a}\overrightarrow{CB}
+ \frac{\log(\tan\frac{\gamma_1}{2}\tan\frac{\alpha_2}{2})}{b}\overrightarrow{AC} = \vec{0} . $$
Next, we claim that
\begin{equation}\label{formulaauxrel}
 {\textstyle\frac{1}{c}}\log(\tan{\textstyle\frac{\alpha_1}{2}}\tan{\textstyle\frac{\beta_2}{2}})
= {\textstyle\frac{1}{a}}\log(\tan{\textstyle\frac{\beta_1}{2}}\tan{\textstyle\frac{\gamma_2}{2}})
= {\textstyle\frac{1}{b}}\log(\tan{\textstyle\frac{\gamma_1}{2}}\tan{\textstyle\frac{\alpha_2}{2}}) .
\end{equation}
To see \eqref{formulaauxrel} one only has to observe $\overrightarrow{AC}=-\overrightarrow{BA}-\overrightarrow{CB}$
and make use of linear independence of $\overrightarrow{BA}$ and $\overrightarrow{CB}$.
If we apply the law of sines and exponentiate \eqref{formulaauxrel}, we will complete the proof of \eqref{formularelation2}.

In order to derive \eqref{formularelation1}, we use trigonometric half-angle formulas, the law of cosines, and some factoring:
$$ \tan^2\frac{\alpha_1}{2} = \frac{1-\cos\alpha_1}{1+\cos\alpha_1}
= \frac{1-(r_A^2\!+\!c^2\!-\!r_B^2)/2r_A c}{1+(r_A^2\!+\!c^2\!-\!r_B^2)/2r_A c}
= \frac{(r_A\!+\!r_B\!-\!c)(r_B\!-\!r_A\!+\!c)}{(r_A\!+\!r_B\!+\!c)(r_A\!-\!r_B\!+\!c)} . $$
Multiplying this one with an analogous expression for $\tan\frac{\beta_2}{2}$ and taking square roots gives
$$ \tan\frac{\alpha_1}{2}\tan\frac{\beta_2}{2} = \frac{r_A+r_B-c}{r_A+r_B+c} , $$
so that \eqref{formulaauxrel} becomes
\begin{equation}\label{formularelwithlogs}
\frac{1}{c}\log\!\Big(\frac{r_A+r_B-c}{r_A+r_B+c}\Big)
= \frac{1}{a}\log\!\Big(\frac{r_B+r_C-a}{r_B+r_C+a}\Big)
= \frac{1}{b}\log\!\Big(\frac{r_C+r_A-b}{r_C+r_A+b}\Big) .
\end{equation}
Exponentiation proves relation \eqref{formularelation1}.
\end{proof}

\section{Cartesian coordinates}

Here we address the problem of determining the coordinates of $P$, given the coordinates of triangle vertices.
The starting point are equalities \eqref{formularelation1}, i.e.\@ their logarithmic version \eqref{formularelwithlogs}.
It is easy to see that these expressions are less than $0$, so it is natural to consider their negatives.
Multiply them further by the semiperimeter $s=\frac{1}{2}(a+b+c)$ of triangle $ABC$ in order to make them ``dimensionless''
and denote the obtained common value by $\lambda$:
$$ -\frac{s}{a}\log\!\Big(\frac{r_B+r_C-a}{r_B+r_C+a}\Big)
= -\frac{s}{b}\log\!\Big(\frac{r_C+r_A-b}{r_C+r_A+b}\Big)
= -\frac{s}{c}\log\!\Big(\frac{r_A+r_B-c}{r_A+r_B+c}\Big) = \lambda . $$
Concentrating on only one expression at a time, we can now write
$$ \frac{r_B+r_C-a}{r_B+r_C+a} = e^{-a\lambda/s} , $$
so that
$$ r_B+r_C = a\,\frac{1+e^{-a\lambda/s}}{1-e^{-a\lambda/s}}
= a\,\frac{e^{a\lambda/2s}+e^{-a\lambda/2s}}{e^{a\lambda/2s}-e^{-a\lambda/2s}} = a\coth{\textstyle\frac{a\lambda}{2s}} $$
and similarly
$$ r_C+r_A = b\coth{\textstyle\frac{b\lambda}{2s}}, \quad r_A+r_B = c\coth{\textstyle\frac{c\lambda}{2s}} . $$
Let us agree to write
\begin{equation}\label{formulauvw}
u=a\coth{\textstyle\frac{a\lambda}{2s}}, \ \ v=b\coth{\textstyle\frac{b\lambda}{2s}}, \ \ w=c\coth{\textstyle\frac{c\lambda}{2s}}
\end{equation}
in all that follows. Hence,
\begin{equation}\label{formulararbrc}
r_A=\frac{1}{2}(v+w-u), \ \ r_B=\frac{1}{2}(w+u-v), \ \ r_C=\frac{1}{2}(u+v-w) .
\end{equation}

Now is the time to observe that the distances $r_A,r_B,r_C$ are not independent.
The simplest equation relating them can be derived from
$$ \textup{area}(\triangle PBC) + \textup{area}(\triangle PCA) + \textup{area}(\triangle PAB) = \textup{area}(\triangle ABC) $$
using Heron's formula:
\begin{align*}
& \sqrt{s_a(s_a-a)(s_a-r_B)(s_a-r_C)} + \sqrt{s_b(s_b-b)(s_b-r_C)(s_b-r_A)} \\
& + \sqrt{s_c(s_c-c)(s_c-r_A)(s_c-r_B)} = \sqrt{s(s-a)(s-b)(s-c)} ,
\end{align*}
with $s_a,s_b,s_c,s$ being semiperimeters of the the four triangles respectively.
Substituting \eqref{formulararbrc}, multiplying by $4$, and simplifying we obtain
\begin{align}
& \sqrt{(u^2-a^2)\big(a^2-(v\!-\!w)^2\big)} + \sqrt{(v^2-b^2)\big(b^2-(w\!-\!u)^2\big)} \nonumber \\
& + \sqrt{(w^2-c^2)\big(c^2-(u\!-\!v)^2\big)} = \sqrt{2(a^2 b^2\!+\!b^2 c^2\!+\!c^2 a^2)-(a^4\!+\!b^4\!+\!c^4)} . \label{formulaeqforlambda}
\end{align}
This is a nonlinear equation for $\lambda$ and then $r_A,r_B,r_C$
are determined by \eqref{formulauvw} and \eqref{formulararbrc}.

It remains to explain how to express coordinates of $P(x_P,y_P)$ from its distances to triangle vertices $A(x_A,y_A)$, $B(x_B,y_B)$, $C(x_C,y_C)$.
Using the formula for Euclidean distance in Cartesian coordinates we get an overdetermined quadratic system for $x_P$ and $y_P$,
\begin{align*}
(x_P-x_A)^2 + (y_P-y_A)^2 & = r_A^2, \\[-1mm]
(x_P-x_B)^2 + (y_P-y_B)^2 & = r_B^2, \\[-1mm]
(x_P-x_C)^2 + (y_P-y_C)^2 & = r_C^2.
\end{align*}
Subtracting the third equation from the first two leads to a linear system
\begin{align*}
2(x_C-x_A)x_P + 2(y_C-y_A)y_P & = x_C^2 - x_A^2 + y_C^2 - y_A^2 + v(w-u), \\
2(x_C-x_B)x_P + 2(y_C-y_B)y_P & = x_C^2 - x_B^2 + y_C^2 - y_B^2 + u(w-v),
\end{align*}
which can be quickly solved as
\begin{align}
x_P & = {\textstyle\frac{(x_A^2+y_A^2-vw)(y_B-y_C)+(x_B^2+y_B^2-wu)(y_C-y_A)+(x_C^2+y_C^2-uv)(y_A-y_B)}
{2x_A(y_B-y_C)+2x_B(y_C-y_A)+2x_C(y_A-y_B)}} , \label{formulacartforx} \\
y_P & = {\textstyle\frac{(x_A^2+y_A^2-vw)(x_B-x_C)+(x_B^2+y_B^2-wu)(x_C-x_A)+(x_C^2+y_C^2-uv)(x_A-x_B)}
{2y_A(x_B-x_C)+2y_B(x_C-x_A)+2y_C(x_A-x_B)}} . \label{formulacartfory}
\end{align}

That way we have established the following theorem.

\begin{theorem}\label{thm2}
Suppose that $P$ is a point inside $\triangle ABC$ satisfying $\vec{E}(P)=\vec{0}$.
Its Cartesian coordinates satisfy \eqref{formulacartforx} and \eqref{formulacartfory},
where $u,v,w$ are defined by \eqref{formulauvw} and $\lambda>0$ is a solution of equation \eqref{formulaeqforlambda}.
\end{theorem}

Turning back to equation \eqref{formulaeqforlambda}, we might want to know the number of its positive solutions.
We claim that the left-hand side is a strictly decreasing function of $\lambda\in(0,\infty)$.
Since
$$ \lambda \,\mapsto\, u^2-a^2 = a^2 \big(\coth^2{\textstyle\frac{a\lambda}{2s}}-1\big) $$
is obviously strictly decreasing, it remains to show that
$$ \lambda \,\mapsto\, |v-w| = \big|b\coth{\textstyle\frac{b\lambda}{2s}}-c\coth{\textstyle\frac{c\lambda}{2s}}\big| $$
increases and that its values stay below $a$. Without loss of generality suppose $b\geq c$.
It is an easy calculus exercise to see that $t\mapsto t\coth t$ is increasing, so the expression inside the last modulus is always positive.
Define
$$ g(t) = b\coth bt - c\coth ct , $$
so that
$$ g'(t) = -\frac{b^2}{\sinh^2 bt} + \frac{c^2}{\sinh^2 ct} . $$
Inequality $g'(t)\geq 0$ is equivalent with
$$ \frac{\sinh bt}{b} \geq \frac{\sinh ct}{c} , $$
which can also be verified easily, using the fact that $t\mapsto(\sinh t)/t$ increases.
Finally, we observe that
$$ \lim_{t\to\infty}g(t) = b-c < a , $$
by the triangle inequality.

Therefore, \eqref{formulaeqforlambda} can have at most one positive solution $\lambda$, which combines nicely with theorem \ref{thm2}
to prove the fact that there can be only one point $P$ inside $T$ such that $\vec{E}(P)=\vec{0}$.
This leads us to the promised result on uniqueness of the maximum point for $V$.

From now on we denote this unique maximum potential point by $P_\textup{max}(x_\textup{max},y_\textup{max})$.
One could name it the \emph{electrostatic center} of $T$, although the term \emph{gravitational center} has already been used
in the literature \cite{h07}, \cite{hmp05} in the study of general convex bodies in $\mathbb{R}^n$.
When we actually want to solve equation \eqref{formulaeqforlambda} for $\lambda$,
we do not know how to do it analytically, so we need to use numerical techniques.
For instance, by taking $A(-1,0)$, $B(2,0)$, and $C(0,2)$ we get
$$ \lambda_{\textup{max}} = 4.010297202743007522718690055346\ldots , $$
and then from \eqref{formulacartforx} and \eqref{formulacartfory},
$$ \begin{array}{rl}x_\textup{max} = & \! 0.272557906914867702024319226991\ldots, \\
y_\textup{max} = & \! 0.704148189723077020171531030875\ldots. \end{array} $$

Even though equation \eqref{formulaeqforlambda} does not seem to be solvable in terms of elementary functions,
we do not really have a rigorous proof of this fact.

\begin{openproblem}
Is it possible to express Cartesian coordinates of $P_\textup{max}$ (or equivalently its parameter $\lambda_\textup{max}$)
as elementary functions of triangle sides $a,b,c$?
\end{openproblem}

If one desires to write the coordinates of $P_\textup{max}$ as explicitly as possible, it will perhaps be easier to do so using a series expansion.
We still require that each term of the series is given by an elementary formula.

\begin{openproblem}
Is it possible to express Cartesian coordinates of $P_\textup{max}$
as two convergent series, $x_\textup{max}=\sum_{n=1}^{\infty}x_n$ and $y_\textup{max}=\sum_{n=1}^{\infty}y_n$,
where both $x_n$ and $y_n$ are elementary functions of $a$, $b$, $c$, and $n$?
\end{openproblem}

Our desire to obtain a series expansion is motivated by a common practice in theoretical physics.
We have to remark once again that numerical schemes for solving \eqref{formulaeqforlambda} actually do lead to
approximations of $x_\textup{max}$ and $y_\textup{max}$ by sequences or series.
However, in that case $(x_n)_{n=1}^{\infty}$ and $(y_n)_{n=1}^{\infty}$ are defined recursively, still without giving us
a single explicit formula that would hold for each $n$.

Equation \eqref{formulaeqforlambda} seems to be a transcendental one, but at least the four square roots can be eliminated by
squaring the equality three times. We do not write down the result of this procedure as it involves more complicated expressions.

\section{Trilinear coordinates}

The point $P_\textup{max}$ deserves to be called a triangle center, as purely physical reasons suggest that it always occupies
the same relative position in any member of a family of mutually similar triangles.
However, the notion of triangle center was rigorously defined in \cite{ck}.
Let us begin by introducing a convenient choice of relative homogeneous coordinates with respect to a given triangle $ABC$.
\emph{Trilinear coordinates} of a point $P$ inside $\triangle ABC$ are any real numbers $\tau_a : \tau_b : \tau_c$
such that
$$ \frac{\tau_a}{d_a} = \frac{\tau_b}{d_b} = \frac{\tau_c}{d_c} , $$
where $d_a,d_b,d_c$ are (directed) distances from $P$ to triangle sides $BC$, $CA$, $AB$ respectively.
Equivalently, $a\tau_a : b\tau_b : c\tau_c$ are the barycentric coordinates of $P$.

A real valued function $f$ defined on the set of all possible triples of triangle side lengths $(a,b,c)$
is called a \emph{triangle center function} if it has the following properties.
\begin{itemize}
\item There exists a real constant $\nu$ such that $f(ta,tb,tc)=t^\nu f(a,b,c)$ for $t>0$, i.e.\@ $f$ is homogeneous of order $\nu$.
\item Equality $f(a,c,b)=f(a,b,c)$ holds for any triple in the domain of $f$.
\item $f$ is not identically $0$.
\end{itemize}
A \emph{triangle center} associated to $f$ is then the point given by trilinear coordinates
\begin{equation}\label{formulacenter}
f(a,b,c) : f(b,c,a) : f(c,a,b) .
\end{equation}
We need to remark that the same center can be associated to many different center functions $f$.

What can we say about our point $P_\textup{max}$?
Calculations from the previous section immediately give
$$ \frac{\tau_a}{\tau_b} = \frac{\textup{area}(\triangle PBC)/a}{\textup{area}(\triangle PCA)/b}
= \frac{\sqrt{((u/a)^2-1)(a^2-(v-w)^2)}}{\sqrt{((v/b)^2-1)(b^2-(w-u)^2)}} , $$
so we see that a good choice of triangle center function for $P_\textup{max}$ is
$$ f(a,b,c) = \sqrt{\Big(\coth^2{\textstyle\frac{a\lambda_\textup{max}}{a+b+c}}-1\Big)
\Big(a^2 - \big(b\coth{\textstyle\frac{b\lambda_\textup{max}}{a+b+c}}-c\coth{\textstyle\frac{c\lambda_\textup{max}}{a+b+c}}\big)^2\Big)} , $$
where $\lambda_\textup{max}$ is the unique positive solution to \eqref{formulaeqforlambda}.
Also, $f$ obviously fulfills all three requirements above (with $\nu=1$).
One only has to observe that $\lambda_\textup{max}$ remains the same if the triangle is scaled by a factor $t$.
This proves the announced assertion that $P_\textup{max}$ is a non-trivial triangle center.

All interesting triangle centers are being collected systematically in C. Kimberling's encyclopedia \cite{ck},
which contains $5622$ entries $X(1)$--$X(5622)$ at the moment of writing of this paper.
Trilinear coordinates are given for these characteristic points, justifying their worth to be mentioned.
In order to detect new centers, the encyclopedia also offers the search among the existing ones using the numerical value of
$$ d_a = \textup{dist}(P,BC) = \frac{2\tau_a \textup{area}(ABC)}{a\tau_a + b\tau_b + c\tau_c} $$
in the particular case of triangle with sides $a=6$, $b=9$, $c=13$.
For point $P_\textup{max}$ it is now easy to compute this value to $30$ decimal digits:
$$ d_a = 2.110731796690289177459836888182\ldots $$
and realize that it does not appear in the list.

Trilinear coordinates for $P_\textup{max}$ are implicit due to the fact that $\lambda_\textup{max}$ is not explicitly given.
Just in the case that the first open problem we stated turns out to have a positive answer, it will be interesting
to see if the trilinear coordinates can be algebraic functions of triangle sides.
Once again we are quite sceptical about that possibility.

\begin{openproblem}
Prove that $P_\textup{max}$ is a transcendental triangle center, i.e.\@ it does not have a trilinear representation \eqref{formulacenter},
with $f$ being an algebraic function of $a,b,c$.
\end{openproblem}

\section{Approximation for the parameter}

It remains to say a few words on estimation of $\lambda_\textup{max}$.
Equation \eqref{formulaeqforlambda} degenerates for an equilateral triangle simply to
$$ 3a^2 \sqrt{\coth^2{\textstyle\frac{\lambda}{3}}-1} = a^2\sqrt{3} , $$
which is easily solved as $\lambda_0=3\log(2+\sqrt{3})$.
An interesting fact we obtained ``experimentally'' is that the exact value of $\lambda_\textup{max}$
for a general triangle $ABC$ is ``quite correlated'' with the quantity
$$ \log\frac{s^2}{27\rho^2} = \log\frac{s^3}{27(s-a)(s-b)(s-c)} \geq 0 , $$
where $\rho$ is radius of the inscribed circle.
Figure \ref{figureratio} sketches graph of the ratio of $\lambda_\textup{max}-\lambda_0$ and this quantity
as a function of two angles $\alpha$ and $\beta$.
\begin{figure}[ht]
\begin{center}
\includegraphics[width=2.5in]{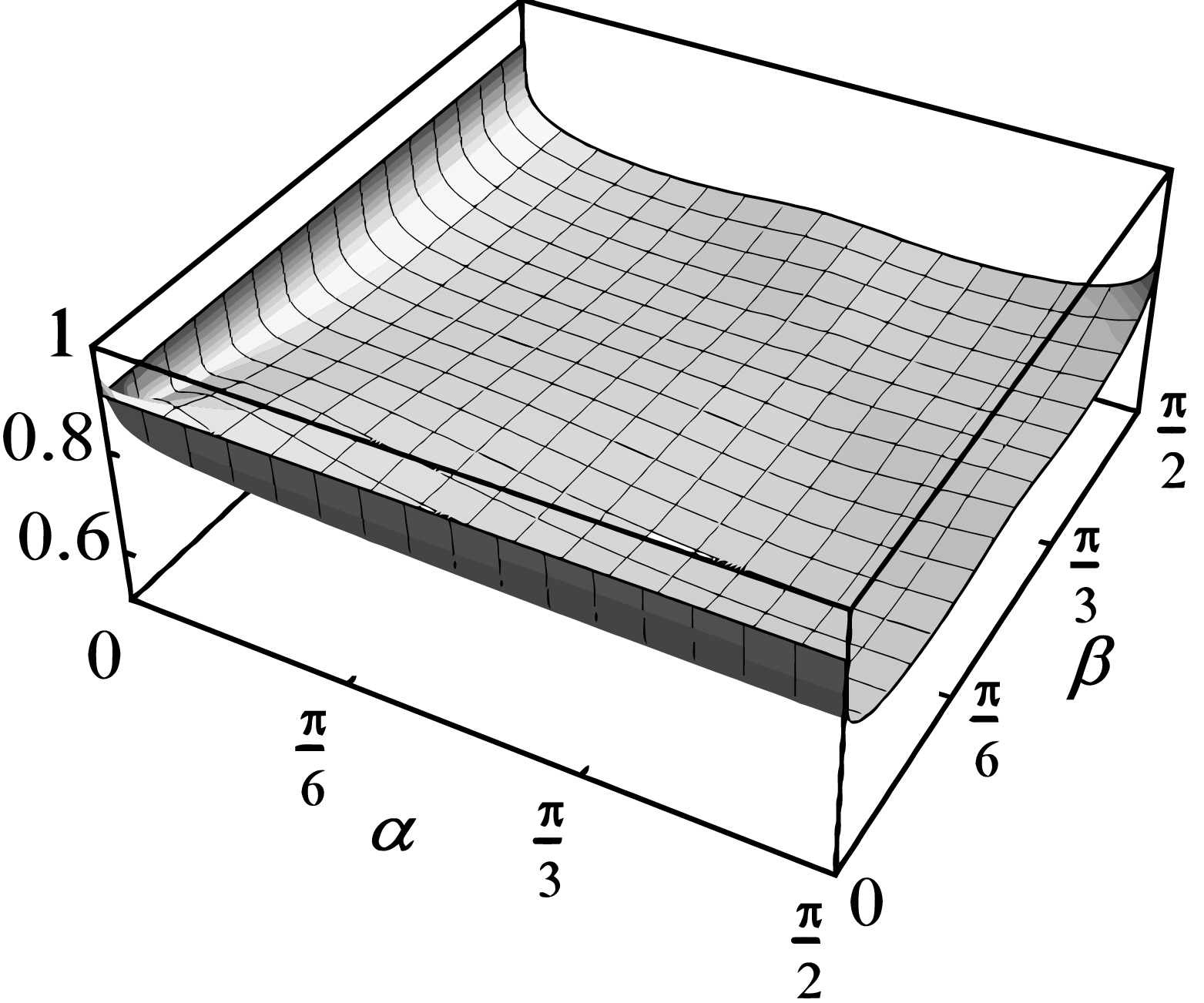}
\end{center}
\caption{Ratio of $\lambda_\textup{max}-\lambda_0$ and $\log(s^2/27\rho^2)$.}
\label{figureratio}
\end{figure}
It is obtained using Plot3D command in Mathematica \cite{m9}.
Note that it is enough to restrict the domain to $0<\alpha,\beta<\pi/2$, because every triangle has at least two acute angles.
The figure illustrates that the ratio is always between (say) $1/2$ and $1$, although we have not established such inequalities rigorously.
The moral of this remark could be that there are some wise choices of the initial approximation to $\lambda$
when solving \eqref{formulaeqforlambda} numerically.

Another interesting observation is related to formulas \eqref{formulauvw}, \eqref{formulacartforx}, and \eqref{formulacartfory}
for Cartesian coordinates of point $P$.
If we now ``free'' the variable $\lambda$ and treat it simply as a parameter that runs over interval $(0,\infty)$,
then the point $P$ traces some planar curve.
Each specific choice of $\lambda$ theoretically corresponds to a triangle center.
It is easy to find limiting positions of $P$ as $\lambda\to 0$ and $\lambda\to\infty$.
These are respectively the centroid $X(2)$ and the incenter $X(1)$.

\section{Related results}

An interesting problem is to investigate extreme points of more general convolution potentials, such as \eqref{formulariesz}
for parameter $p$ taking values other than $-1$, and some of that work has been done by O'Hara \cite{oh12}
following ``experimental speculations'' by Shibata \cite{s09}.
It is well-known that for $p=2$ we obtain the centroid $X(2)$ and the case $p=-2$ will be discussed below.
It seems that all other choices of $p$ lead to unnamed triangle centers and it is not clear which of them satisfy any reasonably nice relations.

Observe that the integral in \eqref{formulariesz} diverges for $p\leq -2$.
One can still define potential difference between two interior points, simply by cutting out small congruent disks around those points.
More precisely, the expression
$$ V_p(P) - V_p(P')
= \iint_{T\setminus\textup{D}_{\varepsilon}(P)} |PQ|^p d\lambda(Q)
- \iint_{T\setminus\textup{D}_{\varepsilon}(P')} |P'Q|^p d\lambda(Q) $$
is well-defined for interior points $P,P'$ and $\varepsilon>0$ small enough, and determines function $V_p$ up to an additive constant.
Our definition is a simpler alternative to the more common approach of subtracting the singular part from the limit as $\varepsilon\to 0$,
as is done in \cite{oh12}.

Let us only comment on the case $p=-2$, as it is also quite interesting and has already appeared in the literature.
Shibata \cite{s09} considered the problem of choosing the position of a street lamp in a triangular park,
in a way that it maximizes the total brightness of the park.
He further reformulates the problem as finding the maximum point of the potential $V_{-2}$ and names it the \emph{illuminating center} of $T$.
Geometrical characterization of such point $P$ inside $\triangle ABC$ that was given in \cite{s09} can be restated as
$$ \frac{\angle BPC}{\textup{area}(\triangle BPC)} = \frac{\angle CPA}{\textup{area}(\triangle CPA)} = \frac{\angle APB}{\textup{area}(\triangle APB)} . $$
Shibata's text does not contain a complete proof of this relation, so let us comment on how one can deduce it rather easily
along the lines of previously presented results.

One can still derive a formula analogous to \eqref{formulaequalszero}. Similarly as in sections
\ref{sectionproperties} and \ref{sectionrelations} we conclude that any stationary point $P$ for $V_{-2}$ in the interior of $T$ now has to satisfy
\begin{equation}\label{formulaequalszero2}
\int_{0}^{2\pi}\! R(\varphi)^{-1} e^{i\varphi} d\varphi = 0 .
\end{equation}
Here we use the same notation as in the proof of theorem \ref{thm1}.
One then calculates
\begin{align*}
& \int_{\varphi_A}^{\varphi_B} R(\varphi)^{-1} e^{i\varphi} d\varphi
= \int_{\alpha_1}^{\pi-\beta_2} \frac{\sin\psi}{d_c} e^{i(\psi+\varphi_A-\alpha_1)} d\psi \\
& = -\frac{ie^{i\varphi_B}}{4r_B} + \frac{ie^{i\varphi_A}}{4r_A} + \frac{e^{i\varphi_B}\cot\beta_2}{4r_B} + \frac{e^{i\varphi_A}\cot\alpha_1}{4r_A}
- \frac{\angle APB}{2i d_c} e^{i\theta_c} ,
\end{align*}
so that \eqref{formulaequalszero2} gives
\begin{align*}
\frac{e^{i\varphi_A}(\cot\alpha_1\!+\!\cot\alpha_2)}{4r_A} + \frac{e^{i\varphi_B}(\cot\beta_1\!+\!\cot\beta_2)}{4r_B}
+ \frac{e^{i\varphi_C}(\cot\gamma_1\!+\!\cot\gamma_2)}{4r_C} & \\
- \frac{\angle BPC}{2i d_a} e^{i\theta_a} - \frac{\angle CPA}{2i d_b} e^{i\theta_b} - \frac{\angle APB}{2i d_c} e^{i\theta_c} & = 0 .
\end{align*}
Straightforward computation shows that the sum of the first three terms is $0$ for just any point $P$, so the above equality becomes
$$ \frac{\angle BPC}{d_a} e^{i\theta_a} + \frac{\angle CPA}{d_b} e^{i\theta_b} + \frac{\angle APB}{d_c} e^{i\theta_c} = 0 , $$
i.e.
$$ \frac{\angle BPC}{\textup{area}(BPC)} \,\overrightarrow{CB} + \frac{\angle CPA}{\textup{area}(CPA)} \,\overrightarrow{AC}
+ \frac{\angle APB}{\textup{area}(APB)} \,\overrightarrow{BA} = \vec{0} . $$
It is easy to fill in the details.

\section{Acknowledgments}
We are grateful to Clark Kimberling for including the point of maximal electrostatic potential
as entry $X(5626)$ in his encyclopedia \cite{ck} shortly after the preprint of this paper became available.
We would also like to thank Dragutin Svrtan for a useful discussion.

\end{document}